\documentclass[12pt,reqno]{amsart}
\usepackage{amsthm,amsfonts,amssymb,euscript}

\newcommand{\bea}{\begin{eqnarray}}
\newcommand{\eea}{\end{eqnarray}}
\def\beaa{\begin{eqnarray*}}
\def\eeaa{\end{eqnarray*}}
\def\ba{\begin{array}}
\def\ea{\end{array}}
\def\be#1{\begin{equation} \label{#1}}
\def \eeq{\end{equation}}

\def\be{{\beta}}

\def\de{\delta}

\def\ep{\epsilon}
\def\eps{\epsilon}

\def\Si{\Sigma}

\def\varep{\varepsilon}

\def\pr{{\partial}}
\def\al{\alpha}
\def\les{\lesssim}

\def\g{\mathbf{g}}

\def\II{{\mathcal I}}

\def\HH{{\mathcal H}}

\def\SS{{\mathcal S}}

\def\D{{\bf D}}

\def\M{{\bf M}}

\def\Z{{\bf Z}}

\def\T{{\bf T}}
\def\E{{\bf E}}

\def\g{{\bf g}}

\def\f12{{\frac 1 2}}

\def\f{\widetilde{f}}

\newtheorem{theorem}{Theorem}[section]
\newtheorem{lemma}[theorem]{Lemma}
\newtheorem{proposition}[theorem]{Proposition}

\numberwithin{equation}{section}

\begin{document}\title[Stationary black holes in vacuum]{Rigidity of stationary black holes  with small angular momentum  
on the horizon}
\author{S. Alexakis}
\address{University of Toronto}
\email{alexakis@math.toronto.edu}
\author{A. D. Ionescu}
\address{Princeton University}
\email{aionescu@math.princeton.edu}
\author{S. Klainerman}
\address{Princeton University}
\email{seri@math.princeton.edu}
\thanks{The first author was partially supported by NSERC grants 488916 and 489103, and a Sloan fellowship. 
The second author was partially supported by
a Packard Fellowship.  The third author was  partially supported by  
NSF grant     DMS-2156449.  All authors were also supported by NSF-FRG grant DMS-1065710.
 }

\begin{abstract}
We prove a black hole rigidity result for slowly rotating stationary solutions of the Einstein vacuum equations. 
More precisely, we prove that the domain of outer communications of a regular stationary vacuum is isometric to the 
domain of outer communications of a Kerr solution, provided that the stationary Killing vector-field $\T$ is  small 
on the bifurcation sphere, i.e.  the corresponding   Black Hole  has  small  angular momentum.
  No other global  restrictions  are necessary. The proof  brings together ideas from   our previous work \cite{AlIoKl} 
with ideas from the classical work of Sudarsky and Wald \cite{Su-Wa} 
on the staticity of stationary black hole solutions 
with zero angular momentum on the horizon.  It is thus  the first  uniqueness    result, in the  framework of smooth, 
asymptotically flat,   stationary solutions,   which combines local considerations near the horizon, via Carleman estimates, 
with information obtained by    global  elliptic   estimates.

 \end{abstract}
 \maketitle
  \tableofcontents

\section{Introduction}\label{introduction}

In this paper we prove a new black hole rigidity result 
for slowly rotating stationary solutions of the Einstein  vacuum  equations. 
More precisely, we show that the domain of outer communications of any smooth  stationary regular vacuum black hole 
with the stationary Killing vector-field $\T$ being small on the 
bifurcation sphere of the horizon must be isometric to the domain of outer communications of a Kerr solution $K(a,M)$ with 
small angular momentum $a/M$.
  This should be compared with  our previous result in \cite {AlIoKl2}  in which  rigidity 
was proved, for the entire range $0\le a <M$, under a  global smallness  assumption on the Mars-Simon tensor associated
 to  the stationary  space-time.  That result rested on three
important  ingredients:
\begin{enumerate}
\item  An unconditional local rigidity result, established in \cite{AlIoKl}  (see also \cite{IoKl3}  )  according to which  a second, rotational Killing vector-field  $\Z$ can be constructed in a small neighborhood of the 
 bifurcate  sphere of the horizon.
 \item An extension argument  for the  Killing vector-field $\Z$ based on    a   global  foliation of the  space-time  with   $\T$- conditional pseudo-convex    hypersurfaces.   
 The crucial $\T$- conditional pseudo-convexity  condition is ensured by the  assumed smallness of the  Mars-Simon tensor.
 \item Once $\Z$ is  globally extended, and thus the space-time is shown  to be  both stationary  and axisymmetric, one can appeal to the classical 
 Carter-Robinson theorem  to conclude the desired  rigidity. 
\end{enumerate} 

The  result we present here is still based on the  first and third ingredients above but replaces  the second  one with  a new ingredient 
inspired  from    the classical work of Sudarsky and Wald \cite{Su-Wa} (see also \cite{Ca2})      on the staticity of stationary, axially 
symmetric\footnote{The result assumes in fact  analyticity   of the space-time  which, according to the well known result of Hawking, 
  implies  axisymmetry.      },  black hole solutions 
with zero angular momentum.  Their result  was based on a   simple  integral formula linking the total extrinsic curvature of   a regular maximal  hypersurface      $\Sigma$ imbedded   in    the space-time and passing through the bifurcate sphere,      with the angular momentum of the horizon.    It can be easily shown\footnote{This step 
 is based on   the assumption of axial symmetry.}           that    zero ADM angular momentum
implies  vanishing angular momentum of the horizon and thus, in view of the above mentioned  formula,   the maximal hyper-surface has to be totally geodesic.  This then implies   the desired conclusion of \cite{Su-Wa}, 
 i.e the space-time  is static. The main observation of our result here  is that    a simple smallness   assumption of $\T$ on the bifurcate sphere\footnote{ This is  equivalent 
with  a   small angular momentum assumption   on the horizon.  It remains open whether this   condition can be replaced with 
  a smallness assumption  of the ADM angular momentum.}  implies  the smallness 
of  the total curvature of the  maximal  hypersurface. This can then be combined   with a  simple application of the  classical    Hopf   Lemma 
to conclude that the entire ergo-region of the black hole     can be covered  by   the   local    neighborhood of the horizon in which the 
second,  rotational, Killing vector-field  $\Z$ has been extended, according to step (1) above.   Away  from the ergo-region $\T$ is time-like   and thus   
 $\T$-conditional  pseudo-convexity  is automatically  satisfied. Thus,  the second Killing  vector-field   $\Z$ can   be easily extended      to the entire space-time  by the results   of \cite{IoKl},        \cite{IoKl2}, \cite{AlIoKl2}.
 Alternatively , since $\T$  is  time-like   in the  complement of   the ergo-region,    the  metric    must   be real analytic in appropriate coordinates, see \cite{Mu}.
  The extension     of $\Z$ can then  be simply  done using  the classical   results  of  \cite{No}.

\subsection{Main Theorem} 
Our result depends on  four   types of assumptions.
\begin{enumerate}
\item    Standard        global  regularity and asymptotic flatness   assumptions  concerning  the stationary space-time  $(\M,\g)$.
\item  Assumptions on the non-degeneracy   of the horizon. 
\item  Assumptions on the existence  and regularity   of an asymptotically  flat  maximal hypersurface passing through the bifurcate sphere  of the horizon.
\item Smallness of the stationary  Killing vectorfield $\T$ on  the bifurcate sphere of the horizon.
\end{enumerate}
\subsubsection{Main Objects}  We  assume that  $(\M,\g)$ is a smooth\footnote{$\M$ is a connected, oriented, time 
oriented, paracompact $C^\infty$ manifold without boundary.} vacuum Einstein space-time of 
dimension $3+1$ and $\T\in T(\M)$ is a smooth Killing vector-field on $\M$. 
We also assume that we are given an embedded partial Cauchy surface
$\Sigma^0\subseteq\M$ and a diffeomorphism $\Phi_0:E_{1/2}\to\Sigma^0$, 
where $E_r=\{x\in\mathbb{R}^3:|x|>r\}$. Moreover, we assume that 
\begin{equation}\label{maxim}
\Sigma_1:=\Phi_0(E_1) \text{ is a maximal hypersurface}.
\end{equation}
The existence of 
asymptoticaly flat maximal surfaces  with $\partial\Sigma_1=\SS_0$ 
in stationary space-times has been derived in \cite{CW} (see theorem 4.2.). 
The required smoothness and decay at spatial infinity that we assume below can be proved by elliptic estimates; these are however 
not the purpose of this paper, so we include them as assumptions.

\subsubsection{Main regularity assumptions}  The regularity assumptions on our space-time, the stationary Killing field, and the bifurcate event horizon 
are precisely as in \cite{AlIoKl2}.
The first assumption is a standard asymptotic flatness
assumption which, in particular,  defines the asymptotic region
$\M^{(end)}$ and the domain of outer communications (exterior
region) $ \E=\II^{-}(\M^{(end)})\cap\II^{+}(\M^{(end)})$. Our  second assumption  concerns the smoothness of the two achronal  boundaries $\delta(\II^{-}(\M^{(end)}))$ in a 
small neighborhood of their intersection $S_0=\de(\II^{-}(\M^{(end)}))\cap \de(\II^{+}(\M^{(end)}))$. Though this second assumption is not directly used here it was  very important in the  local   construction of  a rotational Killing vector-field in \cite{AlIoKl},   see Theorem \ref{old.theorem} below.
\medskip

{\bf{GR.}} (Global regularity assumption)
 We assume that  the restriction of the
 diffeomorphism $\Phi_0$ to $E_{R_0}$,  for $R_0$ sufficiently
large,  extends to a diffeomorphism $\Phi_0:\mathbb{R}\times
E_{R_0}\to \M^{(end)}$, where $\M^{(end)}$ (asymptotic region) is
an open subset of $\M$. In local coordinates $\{x^0,x^i\}$ defined
by this diffeomorphism, we assume that $\T=\partial_0$ and, with
$r=\sqrt{(x^1)^2+(x^2)^2+(x^3)^2}$, that the components of the
space-time metric verify\footnote{\label{foot:asympt}We denote by
 $O_k(r^a)$ any smooth function in  $\M^{(end)}$ which verifies $|\partial^i
f|=O(r^{a-i})$  for any $0\le i\le k$ with $|\partial^i
f|=\sum_{i_0+i_1+i_2+i_3=i}
|\partial_0^{i_0}\partial_1^{i_1}\partial_2^{i_2}\partial_3^{i_3}
f|$.},
\begin{equation}\label{As-Flat}
\g_{00}=-1+\frac{2M}{r}+O_6(r^{-2}),\quad \g_{ij}=\delta_{ij}+O_6(r^{-1}),\quad\g_{0i}=-\ep_{ijk}\frac{2S^jx^k}{r^3}+O_6(r^{-3}),
\end{equation}
for some $M>0$, $S^1,S^2,S^3\in\mathbb{R}$  such that,
\begin{equation}\label{As-Flat2}
J=[(S^1)^2+(S^2)^2+(S^3)^2]^{1/2}\in[0,M^2).
\end{equation}

Let
\begin{equation*}
\E=\II^{-}(\M^{(end)})\cap\II^{+}(\M^{(end)}),
\end{equation*}
where $\II^{-}(\M^{(end)})$, $\II^{+}(\M^{(end)})$ denote the past and
respectively  future
 sets of  $\M^{(end)}$. We assume that $\E$ is globally hyperbolic and
\begin{equation}\label{intersect}
\Sigma^0\cap\II^{-}(\M^{(end)})=\Sigma^0\cap\II^+(\M^{(end)})=\Phi_0(E_1).
\end{equation}
We assume that $\T$ does not vanish at any point of $\E$ and that
every orbit of $\T$ in $\E$ is complete and intersects transversally the
hypersurface $\Sigma^0\cap \E$.
\medskip

{\bf{SBS.}} (Smooth bifurcation sphere assumption) It follows from \eqref{intersect} that
\begin{equation*}
\delta(\II^{-}(\M^{(end)}))\cap\Sigma^0=\delta(\II^+(\M^{(end)}))\cap\Sigma^0=\SS_0,
\end{equation*}
where $\SS_0=\Phi_0(\{x\in\mathbb{R}^3:|x|=1\})$ is an imbedded $2$-sphere (called 
the bifurcation sphere). We assume that there is a neighborhood $\mathbf{O}$ of $\SS_0$ in $\mathbf{M}$ such that the sets
\begin{equation*}
\HH^+=\mathbf{O}\cap \delta(\II^{-}(\M^{(end)}))\quad \text{ and }\quad \HH^-=\mathbf{O}\cap \delta(\II^{+}(\M^{(end)}))
\end{equation*}
are smooth imbedded hypersurfaces. We assume that these hypersurfaces are null, 
non-expanding\footnote{A null hypersurface is said to be non-expanding if the trace of its null 
second fundamental form vanishes identically.}, and intersect transversally in $\SS_0$. Finally, 
we assume that the vector-field $\T$ is tangent to both 
hypersurfaces $\HH^+$ and $\HH^-$, and does not vanish identically on $\SS_0$.
\medskip

\subsubsection{Main Theorem}
\begin{theorem}\label{theorem}
Assume $(\M,\g)$ is a regular black hole exterior satisfying the assumptions {\bf{GR}} and {\bf{SBS}} above, with a stationary Killing field $\T$. 
Assume    in  addition  that the hypersurface $\Sigma_1=\Phi_0(E_1)$ is maximal, and that
\begin{equation}\label{Tsmall}
\|\g(\T,\T)\|_{L^\infty(\SS_0)}<\epsilon^2,
\end{equation}
where $\epsilon$ is a sufficiently small constant\footnote{See subsection \ref{precisesmall} for a precise description of this smallness assumption.}. 
Then  $(M,\g)$ is stationary and axially symmetric,  
thus, in view of the Carter-Robinson theorem \cite{Ca1, Rob, ChCo}, it is isometric 
to a Kerr solution with small angular momentum. 
\end{theorem}

The proof of the theorem consists of two main steps: the first was done in our previous 
work \cite{AlIoKl} where we proved that there exists,  locally  in a small neighborhood of the event horizon,    a second, rotational 
Killing vector field $\Z$  commuting  with $\T$. 
More precisely our result in \cite{AlIoKl} can be stated as follows: 

\begin{theorem}\label{old.theorem}
Under the assumptions above, there exists an open set $\Omega\subset \M$, $\mathcal{S}_0\subseteq\Omega$,  
where $(M,\g)$ admits a second rotational Killing vector-field $\Z$ which commutes with $\T$. 
\end{theorem}

In step  2 we use ideas inspired   from Sudarsky--Wald \cite{Su-Wa} to prove that that $\T$ becomes strictly 
timelike within the set $\Omega\cap\mathbf{E}$ and stays timelike throughout the complement of $\Omega$ in $\mathbf{E}$, provided that the constant $\epsilon$ in \eqref{Tsmall} is sufficiently small. Therefore, see \cite{Mu}, the space-time is analytic in a neighborhood of $\Sigma_1$, outside the domain $\Omega$. Therefore, using \cite{No}, the rotational  Killing vector-field $\Z$ can be extended throughout the exterior region. Alternatively, we can avoid passing through  real analyticity  by 
observing that  the  $\T$- conditional pseudo-convexity is automatically satisfied   if $\T$ is time-like.  We can therefore rely, as well,  on the extension results  proved 
in      \cite{IoKl},\cite{IoKl2},       \cite{AlIoKl2}.

Finally, we can appeal to the Carter-Robinson theorem to conclude that $(\M,\g)$ is isometric to a Kerr solution, 
with a small angular momentum $a=J/M$.

Therefore, the main goal of this paper is to show the following:

\begin{proposition}\label{propo}
Under the above assumptions  
 there exists $\epsilon>0$ sufficiently small such that 
if \eqref{Tsmall} is satisfied then the ergoregion $\{p\in\Sigma_1:\g(\T,\T)\geq 0\}$ of $\T$ is contained in the domain $\Omega$ of
extension of the rotational Killing field $\Z$ guaranteed by Theorem \ref{old.theorem}. 
\end{proposition}

\subsection{The precise quantitative  assumptions}\label{precisesmall}  The assumptions of the previous subsection   were  qualitative in nature.
In this subsection  we make them precise in  a quantitative  way to make sense of our  smallness assumption \eqref{Tsmall}.

 Let $\pr_1,\pr_2,\pr_3$ denote the vectors tangent to $\Sigma^0$, induced by the diffeomorphism $\Phi_0$. Let $\Sigma_r=\Phi_0(E_r)$, where, as before, $E_r=\{x\in\mathbb{R}^3:|x|>r\}$. In particular, for  our original spacelike hypersurface, we have $\Si^0=\Si_{1/2}$.

As in \cite[Section 2.1]{AlIoKl2}, using \eqref{As-Flat} and the assumption that $\Sigma^0$ is spacelike, it follows that there are large constants $A_1$ and $R_1\ge R_0$, such that   $R_1\geq A_1^4$,  with the following properties: on $\Sigma_{3/4}$,
  for any $X=(X^1,X^2,X^3)$,
\begin{equation}\label{prel1}
A_1^{-1}|X|^2\leq\sum_{\al,\be=1}^3X^\al X^\be\g_{\al\be}\leq
A_1|X|^2\quad\text{ and
}\quad\sum_{\al=1}^{3}|\g(\pr_\al,\T)|+|\g(T_0,\T)|\leq A_1.
\end{equation}
In $\Phi_0(\mathbb{R}\times E_{R_1})$, which we continue to denote
by $\M^{(end)}$, $\T=\pr_0$ and  (see notation in footnote
\ref{foot:asympt}),
\begin{equation}\label{prel2}
\begin{split}
\sum_{m=0}^6 r^{m+1}\sum_{j,k=1}^3
|\partial^m(\g_{jk}-\delta_{jk})|&+ \sum_{m=0}^6 r^{m+2}
|\partial^m(\g_{00}+1-2M/r  )|\\
&+\sum_{m=0}^6 r^{m+3} \sum_{i=1}^3
|\partial^m(\g_{0i}+2\ep_{ijk}S^jx^kr^{-3})|\le A_1.
\end{split}
\end{equation}
We construct  a system of coordinates in a small neighborhood $\widetilde{\M}$ of $\Sigma^0\cap\overline{\E}$, which extends both  the
coordinate system of $\M^{(end)}$ in \eqref{prel2} and that of $\Si^0$. We do that with the help of  a smooth vector-field $T'$ which interpolates between $\T$ and $T_0$.
More precisely we  construct  $T'$ in a neighborhood of $\Sigma_{3/4}$ such that $T'=\T$ in $\Phi_0(\mathbb{R}\times E_{2R_1})$ and $T'=\eta(r/R_1)T_0+(1-\eta(r/R_1))\T$ on $\Sigma_{3/4}$, where $\eta:\mathbb{R}\to[0,1]$ is a smooth function supported in $(-\infty,2]$ and equal  to $1$ in $(-\infty,1]$.
  Using now  the flow induced by  $T'$ we extend the original  diffeomorphism $\Phi_0:E_{1/2}\to \Sigma^0$,  to cover a full neighborhood of $\Sigma_1$. Thus  there exists $\varep_0>0$ sufficiently small and a diffeomorphism $\Phi_1:(-\varep_0,\varep_0)\times E_{1-\varep_0}\to\widetilde{\M}$, which agrees with $\Phi_0$ on $\{0\}\times E_{1-\varep_0}\cup (-\varep_0,\varep_0)\times E_{2R_1}$ and  such that $\partial_0=\pr_{x^0}=T'$. 
  
By construction, using also \eqref{prel2} and letting $\varep_0$ sufficiently small depending on $R_1$,
\begin{equation}\label{prel3}
\sum_{j=1}^3|\g_{0j}|+|\g_{00}+1|\leq A_1/(R_1+r)\quad \text{ in }\widetilde{\M}.
\end{equation}
Note, in particular,  that the Killing field $T$ is  time-like on $\Sigma_{R_1}$.

With $\g_{\al\be}=\g(\pr_\al,\pr_\be)$ and $\T=\T^\al\pr_\al$, let
\begin{equation}\label{prel4}
A_2=\sup_{p\in\widetilde{\M}}\sum_{m=0}^6\Big[\sum_{\al,\be=0}^3|\pr^m\g_{\al\be}(p)|+\sum_{\al=0}^3|\pr^m\T^\al(p)|\Big].
\end{equation}
Finally, we fix
\begin{equation}\label{mainconst}
\overline{A}=\max(R_1,A_2,\varep_0^{-1},(M^2-J)^{-1}).
\end{equation}
The constant $\overline{A}$ is our  main effective constant. Notice that this constant also controls the components of the  
contra-variant metric $\g^{\al\be}$ (and their derivatives), as a consequence of \eqref{prel1} and \eqref{prel4}. The constant also
controls the components of the second fundamental form $k$ (see \eqref{df1}) along $\Sigma^0$, and their derivatives
\footnote{We remark that the constant $\overline{A}$ depends only on $M+1/M+1/(M-a)$ in the case when $\E$ is isometric to the 
domain of outer communications of the Kerr space-time $\mathcal{K}(M,a)$. It does not increase when $a$ approaches $0$, 
if $M$ is fixed.}. 

The constant $\epsilon$ in \eqref{Tsmall} will be assumed sufficiently small, depending only on $\overline{A}$. Throughout  the remaining paper we use   the     notation $A\les B$   to denote unequalities    $A\le  C B$,  with  a universal  constant  $  C>0$   which 
 depends only on    $  \overline{A}$. Similarly      $A \gtrsim B$   means   $A\ge C B$.

To summarize, we have  defined a neighborhood $\widetilde{\M}$ of $\Sigma^0\cap\overline{\E}$ and a diffeomorphism $\Phi_1:(-\varep_0,\varep_0)\times E_{1-\varep_0}\to\widetilde{\M}$, $\varep_0>0$, such that the bounds \eqref{prel1}, \eqref{prel2}, \eqref{prel3}, \eqref{prel4} hold (in coordinates induced by the diffeomorphism $\Phi_1$).

\section{Proof of the main theorem}\label{proof}

\subsection{Some useful identities.} In this subsection we gather various formulas relating the Killing vector-field $\T$, the metric $\g$, and 
the maximal hypersurface $\Sigma_1$. Let $h:=\g|_{\Sigma^0}$ denote the induced metric on the hypersurface $\Sigma^0$ and let $\nabla$ denote the induced Levi-Civita connection. 
Also let $T_0$ denote the future unit normal vector-field to $\Sigma^0$. Let $k_{ij}$ denote the second fundamental 
form of the hypersurface $\Sigma^0$,
\begin{equation}\label{df1}
  k(Y, Z):=-\g(\D_Y T_0, Z),
\end{equation}
for all vector-fields $Y, Z$ tangent to $\Si^0$. Notice that
\begin{equation}\label{df2}
\nabla_YZ=\D_YZ +k(Y,Z)T_0, 
\end{equation}
for all vector-fields $Y, Z$ tangent to $\Si^0$. We also recall the Gauss equation,
\begin{equation}\label{df3}
\nabla_ik_{jm}-\nabla_jk_{im}=-\mathbf{R}_{ijm\alpha}(T_0)^\alpha.
\end{equation}

In our case, since $\Sigma_1$ is a maximal hypersurface we have, by definition,
\begin{equation}\label{df4}
h^{ij}k_{ij}=0.
\end{equation}
Using also the Gauss equation and the Einstein vacuum equations it follows that
\begin{equation}\label{df5}
\nabla^ik_{ij}=0.
\end{equation}

We now turn to the a natural decomposition of the Killing vector $\T$ relative to our hypersurface, 
\begin{equation}\label{df6} 
\T=nT_0+X,
\end{equation}
where $n=-\g(\T,T_0)$ (the {\it {lapse function}}) is a smooth real-valued function on $\Sigma^0$, and $X\in T\Sigma^0$ (the {\it{shift vector}}) is a smooth vector-field along $\Sigma^0$ that satisfies $\g(X,T_0)=0$. Since $\T$ is Killing, it follows easily from \eqref{df2} and the decomposition \eqref{df6} that
\begin{equation}\label{df7}
\nabla_iX_j+\nabla_jX_i=2nk_{ij}.
\end{equation}
Finally, the Killing equation together with the maximality imply that $k_{ij}, n$ satisfy the {\it{lapse equation}}
\begin{equation}\label{Lapllapse}
\Delta n=|k|^2n\qquad\text{ along }\Sigma_1,
\end{equation}
where $\Delta:=\nabla^i\nabla_i$ is the Laplace--Beltrami operator induced on the surface $\Sigma^0$.

\subsection{Proof of Proposition \ref{propo}} We show first how to control the lapse function $n$ along the surface $\Sigma_1$. 

\begin{lemma}\label{LemmaP1}
The function $n$ satisfies
\begin{equation}\label{pro0}
n(p)\in[0,1)\qquad\text{ for any }p\in\Sigma_1.
\end{equation}
Moreover, there is a constant $C_1$ that depends only on the constant $\overline{A}$ in \eqref{mainconst} such that
\begin{equation}\label{pro1}
n(p)\geq C_1^{-1}\frac{r(p)-1}{r(p)}\qquad \text{ for any }p\in\Sigma_1.
\end{equation}
\end{lemma} 

\begin{proof}[Proof of Lemma \ref{LemmaP1}] The identity $\T=nT_0+\sum_{j=1}^3X^j\partial_j$ together with the asymptotic flatness assumption \eqref{As-Flat} show that
\begin{equation}\label{pro2}
X^1,X^2,X^3=O_6(r^{-2})\qquad\text{ and }\qquad n=1-M/r+O_6(r^{-2})\qquad\text{ in }\Sigma_{R_1}.
\end{equation}
Moreover, $n\equiv 0$ on $\mathcal{S}_0$, since $\T$ is tangent to $\mathcal{S}_0$. Recall also that $n$ satisfies the elliptic equation $\Delta n=|k|^2n$ along $\Sigma_1$, see \eqref{Lapllapse}. 

The bound \eqref{pro0} follows as a consequence of the weak maximum principle applied to the function $n$ in the domain $\Sigma_1\setminus\overline{\Sigma_{R_1}}$. The bound \eqref{pro1} follows from the proof of the strong maximum principle (Hopf lemma), see, for example,\cite[Chapters 3.1, 3.2]{GiTu}.
\end{proof}

We use now our main assumption \eqref{Tsmall} to show that $k$ is small along $\Sigma_1$.

\begin{lemma}\label{LemmaP2}
For any $i,j\in\{1,2,3\}$ we have
\begin{equation}\label{pro20}
\|k_{ij}\|_{L^\infty(\Sigma_1)}\leq \epsilon^{1/8}.
\end{equation}
\end{lemma}

\begin{proof}[Proof of Lemma \ref{LemmaP2}] We combine the identities \eqref{df5} and \eqref{df7} to derive the formula
\begin{equation}\label{pro21}
n|k|^2=k^{ij}nk_{ij}=\nabla^iX^jk_{ij}=\nabla^i(X^jk_{ij}),
\end{equation}
along $\Sigma_1$. Since $X=\T$ along $\mathcal{S}_0$, it follows from \eqref{Tsmall} that
\begin{equation*}
\sum_{i=1}^3\|X^jk_{ij}\|_{L^\infty(\mathcal{S}_0)}\lesssim\epsilon.
\end{equation*}
Moreover, using the asymptotic flatness assumption \eqref{As-Flat} and the definitions it is easy to see that
\begin{equation*}
k_{ij}=O_6(r^{-2}),\qquad\text{ in }\Sigma_{R_1}\text{ for any }i,j\in\{1,2,3\}.
\end{equation*}
Therefore, we can integrate by parts along the surface $\Sigma_1$ to conclude that
\begin{equation}\label{pro22}
\int_{\Sigma_1}n|k|^2\,d\mu\lesssim\epsilon.
\end{equation}

We can now prove the pointwise bound \eqref{pro20}. Assume, for contradiction, that $|k_{ij}(p)|\geq \epsilon^{1/8}$ for some $i,j\in\{1,2,3\}$ and $p\in\Sigma_1$. In view of the smoothness assumption, it follows that there is a constant $C=C(\overline{A})$ sufficiently large such that $|k_{ij}(p')|\geq \epsilon^{1/8}/2$ for all points $p'\in B(p):=\{p'\in\Sigma_1:|p-p'|\leq C^{-1}\epsilon^{1/8}\}$. In addition, using \eqref{pro1},
\begin{equation*}
\int_{B(p)}n\,d\mu\gtrsim\epsilon^{4/8}.
\end{equation*}
Therefore
\begin{equation*}
\int_{B(p)}n|k|^2\,d\mu\gtrsim\epsilon^{6/8},
\end{equation*}
which contradicts \eqref{pro22}. This proves the desired pointwise bound \eqref{pro20}.
\end{proof}

Finally, we show that $X$ stays small along the surface $\Sigma_1$.

\begin{lemma}\label{LemmaP3}
We have
\begin{equation}\label{pro30}
\sup_{p\in \Sigma_1\setminus\Sigma_{R_1}}\Big[\sum_{j=1}^3|X_j(p)|+\sum_{i,j=1}^3|\nabla_iX_j(p)|\Big]\leq \epsilon^{1/30}.
\end{equation}
\end{lemma}

\begin{proof}[Proof of Lemma \ref{LemmaP3}] We show first that
\begin{equation}\label{pro31}
\|\nabla_lk_{ij}\|_{L^\infty(\Sigma_1)}\leq \epsilon^{1/20}
\end{equation}
for any $l,i.j\in\{1,2,3\}$. Indeed, assume for contradiction that $|\nabla_lk_{ij}(p)|\geq \epsilon^{1/20}$ for some point $p\in \Sigma_1$. Then, using the smoothness assumption and the bound \eqref{pro20}, it follows that $|\partial_lk_{ij}(p')|\geq \epsilon^{1/20}/2$ for all points $p'\in \Sigma_1$ with the property that $|p'-p|\leq \epsilon ^{1/18}$. Therefore there is a point $p'\in\Sigma_1$ with the property that $|p'-p|\leq \epsilon ^{1/18}$ and $|k_{ij}(p')-k_{ij}(p)|\geq \epsilon^{1/9}$. This contradicts the bound \eqref{pro20}.

The vector-field $X$ satisfies the approximate Killing equation
\begin{equation}\label{pro35}
\nabla_iX_j+\nabla_jX_i=2nk_{ij}.
\end{equation}
Recall also that $|X_j(p)|\lesssim\epsilon$ for $p\in\mathcal{S}_0$ and $j\in\{1,2,3\}$. The same interpolation argument as above shows that
\begin{equation}\label{pro36}
\|h(\nabla_VX,\partial_j)\|_{L^\infty(\mathcal{S}_0)}\lesssim\epsilon^{1/3},
\end{equation}
for any $j\in\{1,2,3\}$ and any vector-field $V=V^1\partial_1+V^2\partial_2+V^3\partial_3$ tangent to the bifurcation sphere $\mathcal{S}_0$, satisfying  $\sum_{i=1}^3\|V^i\|_{L^\infty(\mathcal{S}_0)}\leq 1$. Moreover, using \eqref{pro35},
\begin{equation*}
h(\nabla_WX,W)=0\qquad\text{ on }\mathcal{S}_0,
\end{equation*}
since $n$ vanishes along $\mathcal{S}_0$. Combining with \eqref{pro36} it follows that 
\begin{equation*}
\|h(\nabla_{\partial_i}X,\partial_j)\|_{L^\infty(\mathcal{S}_0)}\lesssim\epsilon^{1/3},
\end{equation*}
for any $i,j\in\{1,2,3\}$. Therefore
\begin{equation}\label{pro37}
\sum_{j=1}^3\|X_j\|_{L^\infty(\mathcal{S}_0)}+\sum_{i,j=1}^3\|\nabla_iX_j\|_{L^\infty(\mathcal{S}_0)}\leq \epsilon^{1/4}.
\end{equation}

To prove the desired estimate \eqref{pro30} we need to extend the inequality \eqref{pro37} from the bifurcation sphere $\mathcal{S}_0$ to the region $\Sigma_1\setminus\Sigma_{R_1}$. We use the equation \eqref{pro35}, which is an approximate Killing equation for $X$ along $\Sigma_1$. The argument we present below is a quantitative version of the well-known argument showing that a Killing vector-field vanishes identically in a connected open set if it vanishes up to order $1$ at one point.

More precisely, let 
\begin{equation*}
\pi_{ij}:=\nabla_iX_j+\nabla_jX_i=2nk_{ij},
\end{equation*}
and recall the general formula
\begin{equation*}
\nabla_a\nabla_bX_c=X^dR_{dabc}+(1/2)(\nabla_a\pi_{bc}+\nabla_b\pi_{ac}-\nabla_c\pi_{ab}).
\end{equation*}
Therefore, in view of \eqref{pro20} and \eqref{pro31},
\begin{equation}\label{pro38}
\|\nabla_l\nabla_iX_j-X^dR_{dlij}\|_{L^\infty(\Sigma_1\setminus\Sigma_{R_1})}\lesssim\epsilon^{1/20}.
\end{equation}

Assume now that $p=\Phi_0(r_0\omega)$ is a point in $\Sigma_1\setminus\Sigma_{R_1}$, $r_0\in[1,R_1]$, $\omega\in\mathbb{S}^2$. Let $p'=\Phi_0(\omega)\in\mathcal{S}_0$ and $\gamma:[0,1]\to \Sigma_1\setminus\Sigma_{R_1}$, $\gamma(t)=\Phi_0[(1+(r_0-1)t)\omega]$ denote a curve connecting the points $p'$ and $p$. Let $V(t)=\dot{\gamma}(t)$ denote the vector-field tangent along the curve $\gamma$. In view of \eqref{pro38}, the functions $\nabla_iX_j$ and $X_j$ satisfy the system of transport equations
\begin{equation*}
\begin{split}
&\nabla_VX_j-V^i\nabla_iX_j=0,\\
&\|\nabla_V\nabla_iX_j-X_dV_lR^{dl}_{\,\,\,\,\,\,ij}\|_{L^\infty}\lesssim\epsilon^{1/20}.
\end{split}
\end{equation*}
along the curve $\gamma$. The desired bound \eqref{pro30} follows using also \eqref{pro37}.
\end{proof}

We can now complete the proof of Proposition \ref{propo}

\begin{proof}[Proof of Proposition \ref{propo}] The formula $\T=nT_0+X$ shows that
\begin{equation*}
\g(\T,\T)=-n^2+h(X,X)\qquad\text{ along }\Sigma_1.
\end{equation*}
Using \eqref{pro1} and \eqref{pro30} it follows that 
\begin{equation*}
\g(\T,\T)\leq -\eps\qquad\text{ in }\Sigma_{1+\epsilon^{1/100}}.
\end{equation*}

On the other hand, the main theorem in \cite{AlIoKl} guarantees the existence of a rotational Killing vector-field $\Z$ in a region $\Omega$, which contains the set $\Sigma_1\setminus\Sigma_{1+\rho}$ for some constant $\rho=\rho(\overline{A})>0$. The conclusion of the proposition follows.
\end{proof}

\end{document}